\documentclass{llncs}
\usepackage{amsmath}
\usepackage{amssymb}
\usepackage{enumerate}
\usepackage{varioref}
\usepackage{charter,eulervm}

\title{{\sc A Note on Probe Cographs}}
\author{
 Ton~Kloks 
}
\institute{
 Department of Computer Science\\
 National Tsing Hua University\\ 
 Taiwan
}

\begin{document}

\maketitle

\begin{abstract}
Let $G$ be a graph and let 
$N_1,\dots,N_k$ be $k$ independent sets in $G$. 
The graph $G$ is a $k$-probe cograph if $G$ can be embedded 
into a cograph by adding edges between vertices that are 
contained in the same independent set. 
We show that there exists an $O(k \cdot n^5)$   
algorithm to check if 
a graph $G$ is a $k$-probe cograph. 
\end{abstract}

\section{Introduction}
%%%%%%%%%%%%%%%%%%%%%%

\begin{definition}
A {\em decomposition tree\/} of a graph   
$G$ is a pair $(T,f)$ 
where $T$ is a ternary tree and where $f$ is a bijection 
from the leaves of $T$ to the vertices of $G$. 
\end{definition}

\begin{definition}
A graph is a cograph if it has no induced $P_4$. 
\end{definition}

The class of cographs is exactly the class of graphs which contains 
the one-vertex graph and which is closed under complementation 
and disjoint union. That is, a graph $G$ is a cograph if 
and only if every induced subgraph $H$ of $G$ with 
at least two vertices is either disconnected 
or its complement $\Bar{H}$ is disconnected. 

There are various 
other characterizations. For example, a graph 
$G=(V,E)$ is a cograph if and only if 
every induced subgraph $G^{\prime}=(V^{\prime},E^{\prime})$ 
with at least two vertices 
has a twin, that is, it has two vertices $x$ and $y$ that have the same 
neighbors in $V^{\prime} \setminus \{x,y\}$. 

\bigskip 

Cographs have a decomposition tree which is called a cotree. 
This is a rooted binary tree in which each internal node is 
labeled with $\otimes$ or $\oplus$. The $\otimes$-operator 
connects every vertex in the left subtree with every vertex in 
the right subtree. The $\oplus$-operator unions the two subgraphs 
induced by the left- and right subtree. 

\medskip 

Every edge in the cotree induces a partition of the vertices 
in two parts, say $A$ and $B$, which are the sets of vertices 
that are mapped to the leaves of the two subtrees that 
are separated by the edge. 
The submatrix of the adjacency matrix with rows indexed by the 
vertices of $A$ and the columns indexed by the vertices of 
$B$ has the form 
\begin{equation}
\label{eqn}
\begin{pmatrix} J & 0 \end{pmatrix}
\end{equation}
or the transpose of this. Here $J$ is the all-ones matrix. 
This follows from the fact that the vertices in 
every rooted subtree form a module. 
Cographs can be recognized in linear time. The recognition 
algorithm builds a cotree~\cite{kn:corneil}. 

\bigskip 

Let $G=(V,E)$ be a cograph and let $N_1,\dots,N_k$ be 
$k$ subsets of vertices, not necessarily disjoint. 
Remove all edges $\{x,y\} \in E$ for which there is a subset $N_i$ 
that contains both $x$ and $y$. We call the graphs that are obtained in 
this manner $k$-probe cographs. For the recognition 
problem of $k$-probe cographs we refer to the 
labeled case when sets $N_i$ are a part of the input. 
In that case each vertex has a label which is a 
$0/1$-vector of length $k$ with a 1 in position $i$ 
if the vertex is in $N_i$. 
By Kruskal's theorem~\cite{kn:kruskal}   
$k$-probe cographs are characterized by a finite collection 
of forbidden induced subgraphs (either labeled or unlabeled).   
(It follows that in $k$-probe cographs all induced paths have 
a length  
which is bounded by a function of $k$.) 

\medskip 

By Equation~(\ref{eqn}) $k$-probe cographs have rankwidth 
$k$, since the adjacencies of every vertex across 
a line in the cotree is characterized by its label, which is 
a vector of length $k$. The recognition of  
(labeled or unlabeled) $k$-probe cographs can be expressed 
in monadic second-order logic and it follows that recognizing  
$k$-probe cographs is fixed-parameter 
tractable (see~\cite{kn:hlineny}). 

\medskip 

In this note we show that there is an efficient   
recognition algorithm to recognize labeled $k$-probe cographs. 

\section{Recognition of $k$-probe cographs}
%%%%%%%%%%%%%%%%%%%%%%%%%%%%%%%%%%%%%%%%%%%

\begin{theorem}
There exists an $O(k \cdot n^5)$ algorithm 
for the recognition of labeled $k$-probe cographs. 
\end{theorem}
\begin{proof}
Let $G=(V,E)$ be a labeled graph, that is, every vertex 
$x$ has a label which is a $0/1$-vector of length $k$. 
For $i \in \{1,\dots,k\}$, 
let $N_i$ be the set of vertices that have a 1 in position 
$i$ of their label. 
Then the graph induced by $N_i$ is an independent set 
in $G$. 
The algorithm that we describe below builds a cotree 
for an embedding of $G$ or it concludes that $G$ is not a 
labeled $k$-probe cograph. 

\medskip 

\noindent
Let $X \subseteq V$ be a subset of vertices. Call the set $X$  
a module if every vertex $z \in V \setminus X$ is either 
\begin{enumerate}[\rm (1)]
\item not adjacent to any vertex of $X$, or 
\item $z$ is adjacent to 
 all vertices $x \in X$ of which the label is orthogonal 
to the label of $z$.\footnote{Here, we say that two vectors are 
orthogonal if they don't have a 1 in any common entry. 
In particular, the $\mathbf{0}$-vector is orthogonal to every 
other vector.}   
\end{enumerate}

\medskip 

\noindent
Let $X$ and $Y$ be two disjoint modules. 
Call $X$ and $Y$ twins if 
\begin{enumerate}[\rm (a)]
\item 
\label{case a}
either no vertex of $X$ is adjacent to any vertex of $Y$ 
or every vertex $x \in X$ is adjacent to every vertex $y\in Y$ 
which has a label that is orthogonal to the label of $x$, and 
\item 
\label{case b}
$X \cup Y$ is a module. 
\end{enumerate}

\medskip 

\noindent
Let $X$ and $Y$ be twins. Notice that there is a 
cotree embedding with $X \cup Y$ as a rooted subtree if and only if 
there is a cotree embedding with one of $X$ and $Y$ as a subtree. 

\medskip 

\noindent
The algorithm builds a cotree as follows. Starting with subtrees 
$X$ which consist of one vertex, it grows the subtrees by looking 
for twins. If $\{X,Y\}$ is a twin, the subtrees 
of $X$ and $Y$ are replaced by the subtree for $X \cup Y$. If no 
vertex of $X$ is adjacent to any vertex of $Y$ then the root 
of the subtree for $X \cup Y$ is labeled by $\oplus$, and 
otherwise it is labeled by $\otimes$. 

\medskip 

\noindent
At each stage there is a collection of $O(n)$ feasible 
subtrees. To look for a twin, the algorithm tries all 
pairs. To check if two subtrees $X$ and $Y$ form a twin 
the algorithm checks if either 
no vertex of $X$ is adjacent to any vertex of $Y$, or 
if every vertex $x \in X$ is adjacent to those vertices 
$y \in Y$ of which the label is orthogonal to the label of 
$x$. Furthermore, the algorithm checks 
if $X \cup Y$ is a module. 
Adjacencies are checked in constant time by using the 
adjacency matrix of $G$. To check if 
the labels of two vertices  
are orthogonal takes $O(k)$ time. Thus it can be checked in 
$O(k \cdot n^2)$ 
time if two modules $X$ and $Y$ are twins. 
It follows that within $O(k \cdot n^4)$ 
time either a twin is found or the conclusion is drawn that 
$G$ is not a $k$-probe cograph. Since the cotree has 
$O(n)$ nodes, it follows that a cotree embedding is built in 
$O(k \cdot n^5)$ time, if it exists. 
\qed\end{proof}

\bigskip 

\begin{remark}
For the unlabeled case, the recognition of 
$k$-probe cographs is NP-complete. As noted  
above, the recognition problem is fixed-parameter tractable. 
\end{remark}

\end{document}